\documentclass[12pt]{amsart}

\usepackage{colordvi,amssymb,amsmath,graphicx}

\newtheorem{theorem}{Theorem}[section]
\newtheorem{lemma}[theorem]{Lemma}

\newtheorem{proposition}[theorem]{Proposition}

\newtheorem{rem}[theorem]{Remark}
\newcounter{FNC}[page]
\def\fauxfootnote#1{{\addtocounter{FNC}{2}$^\fnsymbol{FNC}$%
     \let\thefootnote\relax\footnotetext{$^\fnsymbol{FNC}$#1}}}
\newcommand{\dcol}{\Blue}
\newcommand{\demph}[1]{\dcol{{\sl #1}}}
\newcommand{\ba}{{\bf a}}
\newcommand{\bb}{{\bf b}}
\newcommand{\bc}{{\bf c}}
\newcommand{\bff}{{\bf f}}
\newcommand{\bv}{{\bf v}}
\newcommand{\calA}{{\mathcal A}}
\newcommand{\calB}{{\mathcal B}}
\newcommand{\calF}{{\mathcal F}}
\newcommand{\calG}{{\mathcal G}}
\newcommand{\calT}{{\mathcal T}}
\newcommand{\calS}{{\mathcal S}}
\newcommand{\dist}{\mbox{\rm distance}}
\newcommand{\conv}{\mbox{\rm conv}}

\newcommand{\C}{{\mathbb C}}
\renewcommand{\P}{{\mathbb P}}
\newcommand{\R}{{\mathbb R}}
\newcommand{\N}{{\mathbb N}}

\newcommand{\Z}{{\mathbb Z}}

\newcommand{\vol}{\mbox{\rm volume}}
\newcommand{\degree}{\mbox{\rm degree}}

\newcommand{\simplex}{\includegraphics{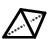}}
\newcommand{\bsimplex}{\includegraphics{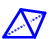}}
\newcommand{\tri}{\includegraphics{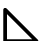}}
\newcommand{\btri}{\includegraphics{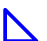}}

\title{Toric degenerations of B\'ezier patches}

\author[L.~Garc{\'\i}a-Puente]{Luis David Garc{\'\i}a-Puente}
\address{Department of Mathematics and Statistics\\
         Sam Houston State University\\
         Huntsville\\
         TX \ 77341\\
         USA}
\email{lgarcia@shsu.edu}
\urladdr{www.shsu.edu/\~{}ldg005}
\author[F.~Sottile]{Frank Sottile}
\address{Department of Mathematics\\
         Texas A\&M University\\
         College Station\\
         Texas \ 77843\\
         USA}
\email{sottile@math.tamu.edu}
\urladdr{http://www.math.tamu.edu/\~{}sottile}
\author[C-G.~Zhu]{Chungang Zhu}
\address{School of Mathematical Sciences\\
          Dalian University of Technology\\
          Dalian 116024
          China}
\email{cgzhu@dlut.edu.cn}
\urladdr{http://gs1.dlut.edu.cn/Supervisor/zhuchungan.page}

\thanks{Work of Garc\'ia-Puente and Sottile supported in part by the Texas Advanced Research
   Program under Grant No. 010366-0054-2007}
\thanks{Research of Sottile supported in part by NSF grant DMS-070105}
\thanks{Research of Zhu supported in part by NSFC grant 10801024, U0935004, and this work
  was conducted in part at Texas A\&M University.}

\PII{}
\copyrightinfo{}{}
\begin{document}
\begin{abstract}
 The control polygon of a B\'ezier curve is well-defined and has geometric significance---there is a
 sequence of weights under which the limiting position of the curve is the control polygon.
 For a B\'ezier surface patch, there are many possible polyhedral control structures, and none 
 are canonical.
 We propose a not necessarily polyhedral control structure for surface patches, regular
 control surfaces, which are certain $C^0$ spline surfaces.
 While not unique, regular control
 surfaces are exactly the possible limiting positions of a B\'ezier patch when the weights are
 allowed to vary. 
\end{abstract}

\maketitle



%
\section{Introduction}

In geometric modeling of curves and surfaces, the overall shape of an individual patch is
intuitively governed by the placement of control points, and a patch may be finely tuned
by altering the weights of the basis functions---large weights pull the patch
towards the corresponding control points.
The control points also have a global meaning as the patch lies within the convex hull of the
control points, for any choice of weights.

This convex hull is often indicated by drawing some edges between the control
points.
The rational bicubic tensor product patches in Figure~\ref{F:bicubic} have the same weights but
different control points, 
and the same $3\times 3$ quadrilateral grid of edges drawn between the control
points.
\begin{figure}[htb]
  \[
   \raisebox{-60pt}{\includegraphics[height=120pt]{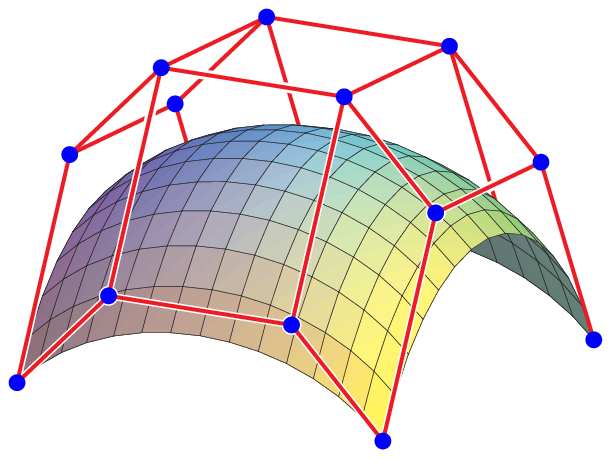}\qquad\qquad
   \includegraphics[height=120pt]{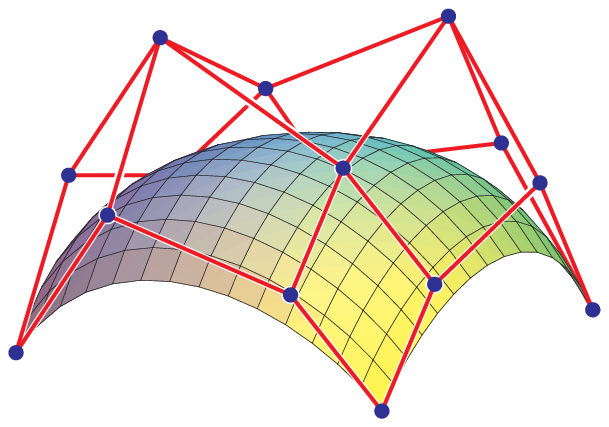} }
  \]
 \caption{Two rational bicubic patches.}
 \label{F:bicubic}
\end{figure}
Unlike the control points or their convex hulls, there is no canonical choice of these edges.
We paraphrase a question posed to us by Carl de Boor and Ron Goldman:
What is the significance for modeling of such control structures (control points plus edges)?

We provide an answer to this question.
These control structures, the triangles, quadrilaterals, and other shapes
implied by these edges, encode limiting positions of the patch when the weights assume extreme
values.
By Theorems~\ref{Th:limit} and~\ref{Th:convex}, the only possible limiting positions of a
patch are the control structures arising from \demph{regular decompositions} (see
Section~\ref{S:regular}) of the points indexing its basis functions and control points,
and any such regular control structure is the limiting position of some sequence of
patches. 

Here are rational bicubic patches with the control points of Figure~\ref{F:bicubic} and 
extreme weights.
\[
  \includegraphics[height=120pt]{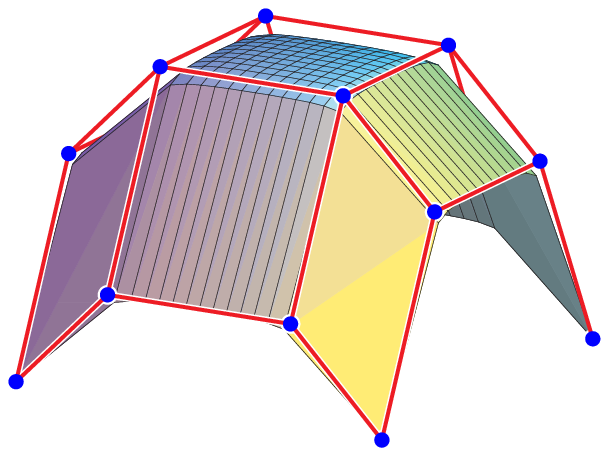}\qquad\qquad
  \includegraphics[height=120pt]{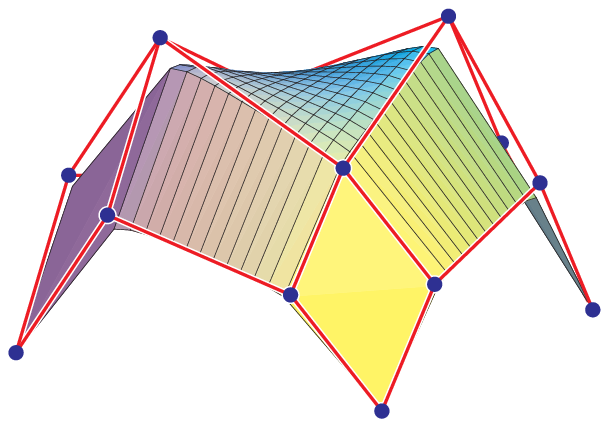}
\]
Each is very close to a composite of nine bilinear tensor product
patches---corresponding to the nine quadrilaterals in their control structures.
The control points of each limiting bilinear patch are the corners of the corresponding
quadrilateral.
These are all planar patches on the left, while 
only the corner quadrilaterals are planar on the right.

The control structure in these examples is a regular decomposition of the
$3\times 3$ grid underlying a bicubic patch.
It is regular as it is induced from the upper convex hull of the graph of a function
on the 16 grid points.
Such a function could be 0 at the four corners, 2
at the four interior points and 1 at the remaining eight edge points.
We show this decomposition on the left below together with an irregular decomposition on the right.
 \begin{equation}\label{Eq:irregular}
     \raisebox{-35pt}{\includegraphics{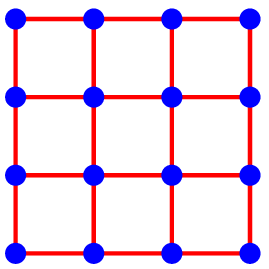}\qquad\qquad
         \includegraphics{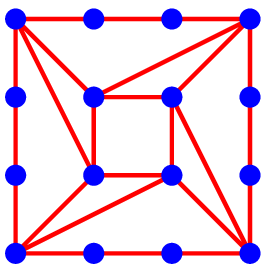}}
 \end{equation}
(If the second decomposition were the upper convex hull of the graph of a function on the grid 
points, and we assume---as we may---that the central square is flat, then the value of the function
at a vertex is lower than the values at a clockwise neighbor, which is impossible outside of Escher
woodcuts.) 

Such control structures and limiting patches were considered in~\cite{Craciun},
but were restricted to triangulations---this restriced the scope of the results.
Our results hold in complete generality and like those of~\cite{Craciun}, rely upon a
construction in computational algebraic geometry called a toric
degeneration~\cite[Ch.~8.3.1]{GKZ}.

While our primary interest is to explain the meaning of control nets for the classical rational 
tensor product patches and rational B\'ezier triangles, we work in the generality of Krasauskas'
toric B\'ezier patches~\cite{Kr02,Kr06}.
The reason for this is simple---any polygon may arise in a regular decomposition of the
points underlying a classical patch.
On the left below is a regular decomposition of the points in the $2\times 2$ grid underlying a
biquadratic patch and on the right is a degenerate patch, which consists of four triangles and
Krasauskas's double pillow.
The pillow corresponds to the central quadrilateral in the $2\times 2$ grid, with the `free' 
internal control point corresponding to the center point of the grid.
 \begin{equation}\label{Eq:pillow}
  \raisebox{-30pt}{\includegraphics{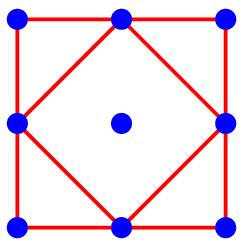}}\qquad\qquad
  \raisebox{-54pt}{\includegraphics[height=120pt]{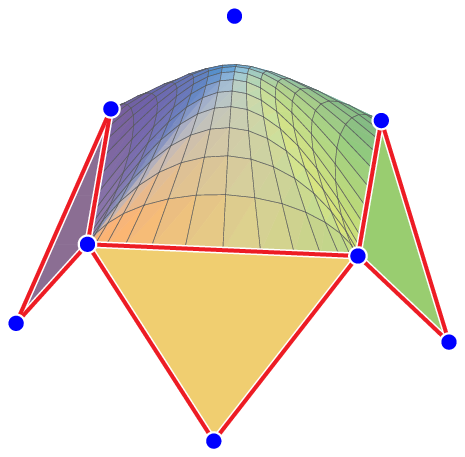}}
 \end{equation}

While our primary interest is in surface patches, our definitions and arguments make sense in
any dimension.
Our discussion will concern surface patches but the proofs in the appendix will be given for patches
of any dimension.

We remark that we do not address the variation diminishing property, which is another fundamental
global aspect of the control polygon of a rational B\'ezier curve.
This states that the total variation of a B\'ezier curve is bounded by the total variation in the
control polygon.
Generalizing this to surfaces is important and interesting, but we currently do not know how to
formulate variation diminishing for general toric patches.

This note is organized as follows. 
We first recall basics of rational B\'ezier triangles and rational tensor product patches and their control nets.  
Next, we present Krasauskas' toric B\'ezier patches and
introduce the crucial notion of a regular polyhedral decomposition. 
In the last section we define the main object in this paper, a regular
control surface, which is a union of toric B\'ezier patches governed
by a regular decomposition. 
We also state our two main theorems, Theorem~\ref{Th:limit}, that regular control surfaces
are limits of toric B\'ezier patches, and Theorem~\ref{Th:convex}, that if a patch is
sufficiently close to a control surface, then that control surface must be regular. 
Proofs appear in the appendix, where we
work in the generality of toric patches in arbitrary dimension. Our
main tools are results of~\cite{KSZ1,KSZ2} which identify all
possible toric degenerations of a projective toric variety.

%
\section{B\'ezier patches and control nets}

We define rational B\'ezier curves and surfaces and tensor product patches in a form that is
convenient for our discussion, and then describe their control nets.
Our presentation differs from the standard formulation~\cite{Fa97} in that
the degree is encoded by the domain.
This linear reparametrization does not affect the resulting curve or surface.
Write $\R_\geq$ for the nonnegative real numbers and $\R_>$ for the positive real
numbers.

Let $d$ be a positive integer.
For each $i=0,\ldots,d$ define the \demph{Bernstein polynomial $\beta_{i;d}(x)$},
\[
     \beta_{i;d}(x)\ :=\ x^i(d-x)^{d-i}\,.
\]
(Substituting $x=dy$ and multiplying by $\tbinom{d}{i}d^{-d}$ for normalization, this
becomes the usual Bernstein polynomial.
This nonstandard presentation, omitting the binomial coefficients, stresses the separate role of
functions and weights.) 
Given weights $w_0,\ldots,w_d\in\R_>$ and control points
$\bb_0,\ldots,\bb_d\in\R^n$ ($n=2,3$), we have the
parameterized \demph{rational B\'ezier curve}
\[
  F(x)\ :=\ \frac{\sum_{i=0}^d w_i\bb_i\beta_{i;d}(x)}%
                 {\sum_{i=0}^d w_i\beta_{i;d}(x)}\ \colon\ [0,d]\ \longrightarrow\ \R^n\,.
\]
Our domain is $[0,d]$ rather than $[0,1]$, for this is the natural
convention for toric patches.

The \demph{control polygon} of the curve is the union of
segments $\overline{\bb_0,\bb_1},\dotsc,\overline{\bb_{d-1},\bb_d}$.
Here are two rational cubic B\'ezier planar curves with their control polygons.
 \begin{equation}\label{Eq:rationalCubicCurves}
  \raisebox{-46pt}{%
   \begin{picture}(130,93)
    \put(10,0){\includegraphics[height=90pt]{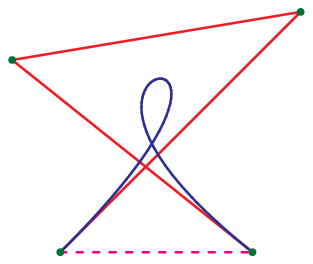}}
    \put(-3,66){$\bb_2$}   \put(118,85){$\bb_1$}
    \put(13, 0){$\bb_0$}   \put(102, 0){$\bb_3$}
   \end{picture}
   \qquad \qquad
   \begin{picture}(130,93)
    \put(11,0){\includegraphics[height=90pt]{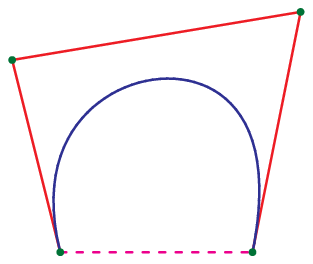}}
    \put(-1,66){$\bb_1$}   \put(119,85){$\bb_2$}
    \put(14, 0){$\bb_0$}   \put(103, 0){$\bb_3$}
   \end{picture}%
  }
 \end{equation}

There are two standard ways to extend this to surfaces. 
The most straightforward
gives rational tensor product patches. Let $c,d$ be positive integers and for
each $i=0,\dotsc,c$ and $j=0,\dotsc,d$ let $w_{(i,j)}\in\R_>$ and
$\bb_{(i,j)}\in\R^3$ be a weight and a control point. 
The associated rational tensor product patch of bidegree $(c,d)$ is the image of the map
$[0,c]\times[0,d]\to\R^3$,
\[
  F(x,y)\ :=\ 
    \frac{\sum_{i=0}^c\sum_{j=0}^d w_{(i,j)}\bb_{(i,j)}\beta_{i;c}(x)\beta_{j;d}(y)}%
         {\sum_{i=0}^c\sum_{j=0}^d w_{(i,j)} \beta_{i;c}(x)\beta_{j;d}(y)}\ .
\]
Triangular B\'ezier patches are another extension.
Set
\[
   \Blue{d}\btri\ :=\ \{(x,y)\in\R^2\mid 0\leq x,y\mbox{ and }x+y\leq d\}
\]
and set $\calA:=d\tri\cap\Z^2$, the points with integer coordinates (lattice points) in
the triangle $d\tri$. 
For $(i,j)\in\calA$, we have the bivariate Bernstein polynomial
\[
   \Blue{\beta_{(i,j);d}(x,y)}\ :=\ x^iy^j(d-x-y)^{d-i-j}\,.
\]
Given weights $w=\{w_{(i,j)}\mid(i,j)\in\calA\}$ and control points
$\calB=\{\bb_{(i,j)}\mid(i,j)\in\calA\}$, the associated triangular 
rational B\'ezier patch is the image of the map $d\tri\to\R^3$,
\[
  F(x,y)\ :=:\ 
    \frac{\sum_{(i,j)\in\calA}w_{(i,j)}\bb_{(i,j)}\beta_{(i,j);d}(x,y)}%
         {\sum_{(i,j)\in\calA}w_{(i,j)}\beta_{(i,j);d}(x,y)}\ .
\]

The control points of a B\'ezier curve are naturally connected in
sequence to give the control polygon, which is a piecewise linear
caricature of the curve. 
For a surface patch there are however many ways to interpolate the control points by edges to form a
control net. 
There also may not be a way to fill in these edges with polygons to form a polytope. 
Even when this is possible, the significance of this structure for the shape of the patch is not
evident, except in special cases. 
For example, if the control points are the graph of a convex function over the lattice points, then the
patch is convex~\cite{Da91,DM88}. 
For such control points, the obvious net consists of the upward-pointing facets of the convex hull
of the graph.  
This is the net of the bicubic patches in Figure~\ref{F:bicubic}.

%
\section{Toric patches and toric varieties}
\label{S:toric}

Krasauskas's toric patches~\cite{Kr02} are a natural extension of rational B\'ezier triangles and
rational tensor product patches to arbitrary polygons whose vertices have integer
coordinates, called \demph{lattice polygons}.
They are based on toric varieties~\cite{CLS,Fu93} from algebraic geometry which get their name as
they are natural compactifications of algebraic tori $(\C^*)^n$, where
$\dcol{\C^*}:=\C\setminus\{0\}$.
They are naturally associated to lattice polygons (and
in higher dimensions, lattice polytopes), and the positive real
part~\cite[Ch.~4]{Fu93}~\cite{So03} of a toric variety is canonically identified with the
corresponding polygon/polytope.

Toric patches  begin with a finite set $\calA\subset\Z^2$ of (integer) lattice points.
The convex hull of $\calA$ is the set of all convex combinations
\[
   \sum_{\ba\in\calA} p_\ba \ba
    \qquad\mbox{where}\qquad
   p_\ba\geq 0
    \qquad\mbox{and}\qquad
   1\ =\ \sum_{\ba\in\calA} p_\ba\,
\]
of points of $\calA$, which is a lattice polygon and is written \dcol{$\Delta_\calA$}.
To each edge $e$ of $\Delta_\calA$, there is a valid inequality
$h_e(x,y)\geq0$ on $\Delta_\calA$, where $h_e(x,y)$ is a linear polynomial with integer
coefficients having no common integer factors that vanishes on the edge $e$.
For example, if $\calA=d\tri\cap\Z^2$, then
$\Delta_\calA=d\tri$ and the inequalities are
\[
   x\geq0,\quad y\geq0,\quad\mbox{and}\quad d{-}x{-}y\geq 0.
\]
For another example, the central quadrilateral of~\eqref{Eq:pillow} has inequalities
\[
  x{+}y{-}1\geq0,\,\
  1{+}x{-}y\geq0,\,\
  3{-}x{-}y\geq0,\,\ \mbox{and}\,\
  1{+}y{-}x\geq0.
\]
Let $E$ be the set of edges of the polygon $\Delta_\calA$.
To each lattice point $\ba\in\calA$, define the
\demph{toric basis function} $\beta_{\ba,\calA}\colon\Delta_\calA\to\R$ to be
\[
  \beta_{\ba,\calA}(x,y)\ :=\ \prod_{e\in E} h_e(x,y)^{h_e(\ba)}.
\]
This is strictly positive in the interior of $\Delta_\calA$.
If $\ba$ lies on an edge $e$ of $\Delta_\calA$, then $\beta_{\ba,\calA}$ is strictly
positive on the relative interior of $e$, and if $\ba$ is a vertex, then
$\beta_{\ba,\calA}(\ba)>0$.
In particular the toric basis functions have no common zeroes in $\Delta_\calA$.

Observe that the toric basis functions for
$\calA=[0,c]\times[0,d]\cap\Z^2$ and $\calA=d\tri\!\cap\Z^2$
are equal to the Bernstein polynomials $\beta_{i;c}(x)\beta_{j;d}(y)$ and
$\beta_{(i,j);d}(x,y)$ underlying the tensor product and triangular B\'ezier patches.

Toric patch also require \demph{weights} and control points.
Let $\#\calA$ be the number of points in $\calA$.
Let $\R_>^\calA$ be $\R_>^{\#\calA}$ with coordinates
$(z_\ba\in\R_>\mid\ba\in\calA)$ indexed by elements of $\calA$. A
toric B\'ezier patch of shape $\calA$ is given by a collection of positive weights
$w=(w_\ba\mid \ba\in\calA)\in\R_>^\calA$ and control points
$\calB=\{\bb_\ba\mid\ba\in\calA\}\subset\R^3$.
These define a map $\Delta_\calA\to\R^3$,
 \begin{equation}\label{Eq:TBP}
   F_{\calA,w,\calB}(x,y)\ :=\ 
        \frac{\sum_{\ba\in\calA} w_\ba \bb_\ba \beta_{\ba,\calA}(x,y)}%
               {\sum_{\ba\in\calA} w_\ba \beta_{\ba,\calA}(x,y)}\ .
 \end{equation}
Since the toric basis functions are nonnegative on $\Delta_\calA$ and have no common
zeroes, this denominator is strictly positive on $\Delta_\calA$.
Write \demph{$Y_{\calA,w,\calB}$} for the image of $\Delta_\calA$ under the map $F$, which
is a \demph{toric B\'ezier patch} of shape $\calA$.

The map $F\colon\Delta_\calA\to\R^3$ admits a factorization
 \begin{equation}\label{Eq:factorization}
   F\ \colon\ \Delta_\calA\ \xrightarrow{\ \beta_{\calA}\ }\ 
    \simplex^{\calA}\ \xrightarrow{\ w\cdot \ }\ 
    \simplex^{\calA}\ \xrightarrow{\ \pi_\calB\ }\ \R^3\,,
 \end{equation}
where $\simplex^{\calA}\subset\R^\calA$ is the standard simplex of dimension $\#\calA-1$, the
map $\beta_{\calA}$ is given by the toric basis functions $\beta_{\ba,\calA}$, the map $w\cdot$ is
(essentially) coordinatewise multiplication by the weights $w$, and the map $\pi_\calB$ is
a projection given by the control points $\calB$.
The purpose of this factorization is to clarify the role of the weights in a toric patch
by isolating their effect.
The image $\beta_{\calA}(\Delta_\calA)\subset \simplex^{\calA}$ is a standard toric
variety $X_\calA$.
Acting on this by the map $w\cdot$ gives a translated toric variety
$X_{\calA,w}$, which we call a \demph{lift} of the patch $Y_{\calA,w,\calB}$ as its image
under the projection $\pi_\calB$ is $Y_{\calA,w,\calB}$. 
We use results on the limiting position of the translates
$X_{\calA,w}$ as the weights are allowed to vary, which are called toric degenerations.

We make this precise.
Let $\R_\geq^\calA$ be $\R_\geq^{\#\calA}$ with coordinates
$(z_\ba\in\R_\geq\mid\ba\in\calA)$ indexed by elements of $\calA$.
The standard simplex
 \[
  \bsimplex^{\calA}\ :=\
    \{ z\in\R_\geq^\calA\mid {\textstyle \sum_{\ba\in\calA}z_\ba=1}\}\ 
 \]
is the convex hull of the standard basis in $\R^\calA$.
It has homogeneous coordinates,
 \begin{equation}\label{Eq:Homog_Coords}
  [z_\ba\mid\ba\in\calA]\ :=\
    \frac{1}{\sum_{\ba\in\calA}z_\ba}(z_\ba\mid\ba\in\calA)\,.
 \end{equation}
Geometrically, $[z_\ba\mid \ba\in\calA]\in\simplex^{\calA}$ is the unique point where the ray
$\R_>\cdot(z_\ba\mid \ba\in\calA)$ meets the simplex $\simplex^{\calA}$.

Let $\beta_\calA\colon\Delta_\calA\to\simplex^{\calA}$ be the map
$\beta_\calA(x,y)=[\beta_{\ba,\calA}(x,y)\mid \ba\in\calA]$
Given weights $w\in\R^\calA_>$, we have the map
$w\cdot \colon\simplex^{\calA}\to\simplex^{\calA}$ defined by
\[
   w\cdot [z_\ba\mid\ba\in\calA]\ =\ [w_\ba z_\ba\mid \ba\in\calA]\,.
\]
Lastly, given control points $\calB$, define the linear map
$\pi_\calB\colon\R^\calA\to\R^3$ by
\[
  \pi_\calB(z)\ :=\ \sum_{\ba\in\calA} \bb_\ba z_\ba\,.
\]
The image of the simplex $\simplex^\calA$ under $\pi_\calB$ is
the convex hull of the control points $\calB$, and by these definitons,
the map $F$ in~\eqref{Eq:TBP} defining the toric B\'ezier patch is the
composition~\eqref{Eq:factorization}.

We call $Y_{\calA,w,\calB}$ a toric patch because the image 
$\beta_\calA(\Delta_\calA)$ is a toric variety, which we now explain.
Elements $\ba$ of $\Z^2$ are exponents of monomials,
\[
   \ba\ =\ (a,b)\ \longleftrightarrow\ x^ay^b\,,
\]
which we will write as $x^{\ba}$.
The points of $\calA$ define a map
$\varphi_{\calA}\colon\R_>^2\to\simplex^\calA$ by
\[
  \varphi_{\calA}(x,y)\ :=\ [x^\ba\mid \ba\in\calA]\,.
\]
The closure in $\simplex^\calA$ of the image of $\varphi_{\calA}$
is the  \demph{toric variety $X_{\calA}$}.
We have the following result of Krasauskas~\cite{Kr02}.

\begin{proposition}
  The image of $\Delta_\calA$ under the map $\beta_{\calA}$ is the toric variety
  $X_{\calA}$.
\end{proposition}

Toric patches share with rational B\'ezier patches the following recursive
structure. If $\ba$ is a vertex of $\Delta_\calA$, then
$\bb_\ba=F_{\calA,w,\calB}(\ba)$ is a point in the patch. 
If $e$ is the edge between two vertices of $\Delta_\calA$, then the restriction 
$F_{\calA,w,\calB}|_{e}$ of $F_{\calA,w,\calB}$ to $e$ is the 1-dimensional toric patch given by 
the points of $\calA$ lying on $e$ and the corresponding weights, which is a rational B\'ezier
curve.
For example, the edges of the rational bicubic
patches in Figure~\ref{F:bicubic} are all rational cubic B\'ezier curves.

%
\section{Regular polyhedral decompositions}
\label{S:regular}

We recall the definitions of regular (or coherent) polyhedral subdivisions from geometric
combinatorics, which were introduced in~\cite[\S~7.2]{GKZ}.
Because subdivision has a different meaning in modeling, we instead use the term \demph{decomposition}.
Let $\calA\subset\R^2$ be a finite set and suppose that
$\lambda\colon\calA\to\R$ is a function.
We use $\lambda$ to lift the points of $\calA$ into $\R^3$.
Let $P_\lambda$ be the convex hull of the lifted points,
\[
  P_\lambda=\conv\{(\ba,\lambda(\ba))\mid\ba\in\calA\}\subset\R^3.
\]
Each face of $P_\lambda$ has an outward pointing normal vector, and
its \demph{upper facets} are those whose normal has positive last
coordinate. 
If we project these upper facets back to $\R^2$, they cover the polygon $\Delta_\calA$ and are the
facets of the \demph{regular polyhedral decomposition $\calT_\lambda$} of $\Delta_\calA$ induced by
$\lambda$.  
(Lower facets also induce a regular polyhedral subdivision, which equals $\calT_{-\lambda}$, and
so it is no loss of generality to work with upper facets.)

The edges and vertices of  $\calT_\lambda$ are the images
of the edges and vertices lying on upper facets. 
Here are the upper facets and regular polyhedral decompositions given by two different lifting functions
for the points $\calA$ underlying a biquadratic tensor product patch.
\[
  \includegraphics[height=120pt]{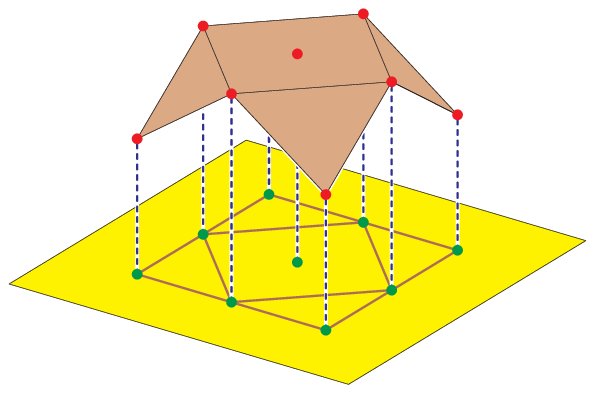}\qquad\quad
  \includegraphics[height=120pt]{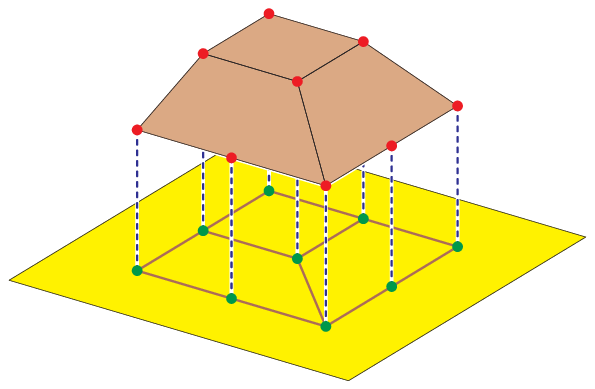}
\]
More generally, a \demph{polyhedral decomposition} of $\Delta_\calA$ is a
collection \dcol{$\calT$} of polygons, line segments, and points of $\calA$,
whose union is $\Delta_\calA$, where any edge, vertex, or endpoint of a segment
also lies in $\calT$, and any two elements of $\calT$ are either disjoint or their
intersection is an element of $\calT$.
A decomposition $\calT$ is \demph{regular} if it is induced from a lifting function.

A \demph{decomposition $\calS$} of the configuration $\calA$ of points is a collection
$\calS$ of subsets of $\calA$ called \demph{faces}. 
The convex hulls
of these faces are required to be the polygons, line segments, and vertices of a
polyhedral decomposition \demph{$\calT(\calS)$} of $\Delta_\calA$. In particular, the
intersection of any face with the convex hull $\Delta_{\calF}$ of
another face $\calF$ of $\calS$ is either empty, a vertex of $\Delta_{\calF}$,
or the points of $\calF$ lying in some edge of $\Delta_{\calF}$.
A face $\calF$ is a \demph{facet}, \demph{edge}, or \demph{vertex} of $\calS$ as its
convex hull $\Delta_{\calF}$ is a polygon, line segment, or vertex.
The decomposition $\calS$ is \demph{regular} if the polyhedral decomposition $\calT(\calS)$ is
regular.

Below are two different lifting functions that induce the same
regular polyhedral decomposition of the $2\times 2$ square underlying a biquadratic patch, but
different regular decompositions of $\calA$.
\[
  \includegraphics[height=120pt]{figures/Upperhull.1.eps}\qquad\quad
  \includegraphics[height=120pt]{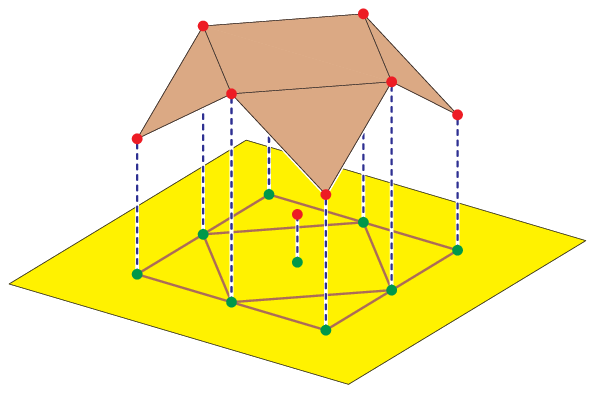}
\]
The center point of $\calA$ does not lie in any face of the
decomposition on the right as its lift does not lie on any upper facet.

Here is a one-dimensional example.
Let $\lambda$ take the values $\{0,1,2,0\}$ on the points $\{0,1,2,3\}$ underlying rational cubic
  B\'ezier curves.
This induces a regular decomposition of $\{0,1,2,3\}$ with facets
 \begin{equation}\label{Eq:cubic_decomposition}
   \{0,1,2\}\qquad\mbox{and}\qquad \{2,3\}\,.
 \end{equation}
%

%
\section{Regular control surfaces}
\label{S:control_surface}

Regular control surfaces are possible limiting positions of patches.
We first illustrate these notions on a rational cubic curve in the plane.
The curves of~\eqref{Eq:rationalCubicCurves} have weights $(1,4,4,1)$ on the Bernstein polynomials
$\beta_{0;3},\beta_{1;3},\beta_{2;3},\beta_{3;3}$, respectively. 
The lifting function inducing the decomposition~\eqref{Eq:cubic_decomposition} gives a family of
weights $(1,4t,4t^2,1)$ as $t\in\R_>$ varies.
When $t=5$, and we use the control points of~\eqref{Eq:rationalCubicCurves}, we get the following
curves. 
 \begin{equation}\label{Eq:rationalCubicCurves_degen}
  \raisebox{-46pt}{%
  \begin{picture}(130,93)
    \put(10, 0){\includegraphics[height=90pt]{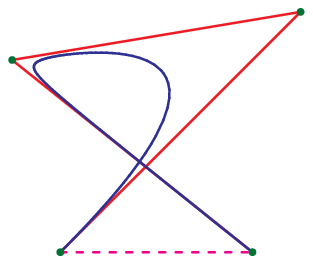}}
    \put(-3,66){$\bb_2$}   \put(118,85){$\bb_1$}
    \put(13, 0){$\bb_0$}   \put(102, 0){$\bb_3$}
   \end{picture}
   \qquad \qquad
   \begin{picture}(130,93)
    \put(11, 0){\includegraphics[height=90pt]{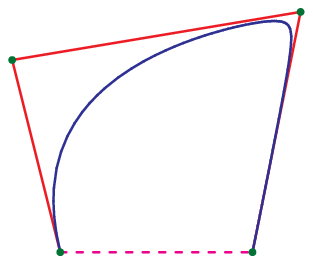}}
    \put(-1,66){$\bb_1$}   \put(119,85){$\bb_2$}
    \put(14, 0){$\bb_0$}   \put(103, 0){$\bb_3$}
  \end{picture}%
  }
 \end{equation}
To consider the limit as $t\to\infty$, write the Bernstein polynomials in
homogeneous form as $\beta_{i;3}:=u^iv^{3-i}$ for $i=0,\dotsc,3$, for then the 
cubic curve is the image of points $(u,v)\in(\R_>)^2$.

Limiting positions are given by restrictions to the facets of the
decomposition~\eqref{Eq:cubic_decomposition}. 
Multiplying the $\beta_{i;3}$ by the weights and restricting to each facet,
we get basis functions 
\[
  \{ v^3, 4tv^2u, 4t^2vu^2\}\,,
   \qquad\mbox{and}\qquad
  \{4t^2v u^2, u^3\}\,.
\]
These give rational B\'ezier curves
\[
  \frac{v^3\bb_0+4tv^2u\bb_1+4t^2vu^2\bb_2}{v^3+4tv^2u+4t^2vu^2}
   \qquad\mbox{and}\qquad
  \frac{4t^2vu^2\bb_2+u^3\bb_3}{4t^2vu^2+u^3}\ .
\]
Dividing out the common factor of $v$ from the first and replacing $tu$ by $u$, and similarly
dividing out $u^2$  from the second and replacing $vt^2$ by $v$, we get
\[
  \frac{v^2\bb_0+4vu\bb_1+4u^2\bb_2}{v^2+4vu+4u^2}
   \qquad\mbox{and}\qquad
  \frac{4v\bb_2+u\bb_3}{4v+u}\ ,
\]
which are rational quadratic and linear B\'ezier curves.
With the control points of~\eqref{Eq:rationalCubicCurves_degen} these are,
 \begin{equation}\label{Eq:rationalCubicCurves_CS}
  \raisebox{-46pt}{%
  \begin{picture}(130,93)
    \put(10, 0){\includegraphics[height=90pt]{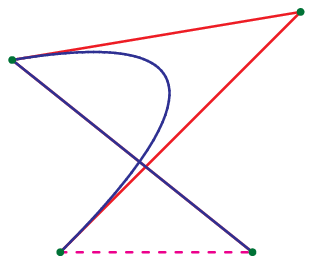}}
    \put(-3,66){$\bb_2$}   \put(118,85){$\bb_1$}
    \put(13, 0){$\bb_0$}   \put(102, 0){$\bb_3$}
   \end{picture}
   \qquad \qquad
   \begin{picture}(130,93)
    \put(11, 0){\includegraphics[height=90pt]{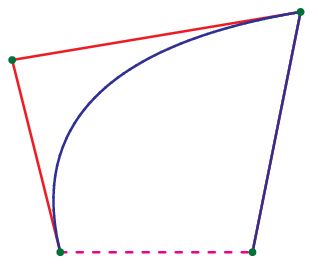}}
    \put(-1,66){$\bb_1$}   \put(119,85){$\bb_2$}
    \put(14, 0){$\bb_0$}   \put(103, 0){$\bb_3$}
   \end{picture}%
  }
 \end{equation}
These are regular control surfaces of the corresponding curves induced by the
decomposition~\eqref{Eq:cubic_decomposition}.
We explain this in general.\smallskip

Let $\calA\subset\Z^2$ be a finite set, $w\in\R^\calA_>$ be weights
and $\calB=\{\bb_\ba\mid \ba\in\calA\}$ be control points for a
toric patch $Y_{\calA,w,\calB}$ of shape $\calA$.

Suppose that we have a decomposition $\calS$ of $\calA$. 
We may use the weights $w$ and control points $\calB$
indexed by elements of a facet $\calF$ as weights and control points for a
toric patch of shape $\calF$, written \dcol{$Y_{\calF,w|_\calF,\calB|_\calF}$}.
In fact, this can be done for any face of $\calS$.
The union
\[
  \dcol{Y_{\calA,w,\calB}(\calS)}\ := \
  \bigcup_{\calF\in\calS} Y_{\calF,w|_\calF,\calB|_\calF}\,,
\]
of these patches is the \demph{control surface} induced by the decomposition $\calS$.
As the domain of a patch of shape $\calF$ is the convex hull $\Delta_{\calF}$ of $\calF$ and faces
of toric patches are again toric patches, 
the control surface of a patch induced by a decomposition is naturally a
$C^0$ spline surface.
A control surface $Y_{\calA,w,\calB}(\calS)$ is \demph{regular} if the decomposition $\calS$ is
regular.

Here are the control surfaces of the bicubic patches from Figure~\ref{F:bicubic}.
\[
  \includegraphics[height=120pt]{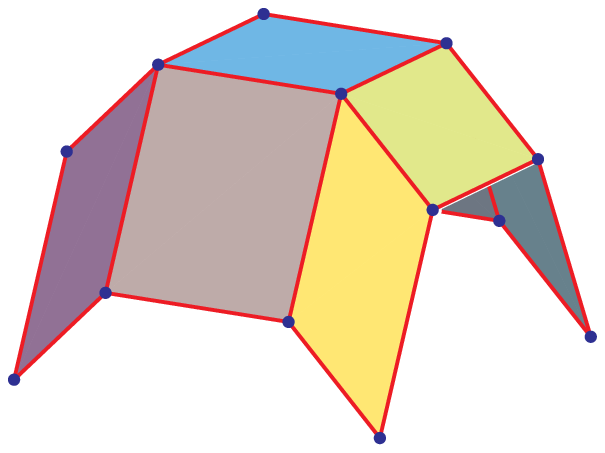}\qquad\qquad
  \includegraphics[height=120pt]{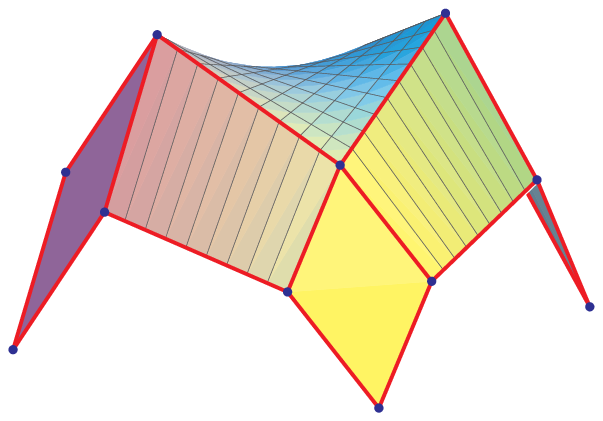}
\]
These control surfaces are regular as they are induced by the $3\times 3$ grid, which is a
regular decomposition.
Below is the irregular decomposition of the $3\times 3$ grid from~\eqref{Eq:irregular} and a
corresponding irregular control surface.
 \begin{equation}\label{Eq:nonConvex}
  \raisebox{-35pt}{%
   \begin{picture}(88,75)(-12,0)
     \put(0,0){\includegraphics{figures/3x3nonConvex.eps}}
     \put(32,34){$A$}   \put(43,12){$C$} 
      \put(-12,16){$B$} \put(-2,20){\vector(1,0){30}} 
     \put(0,-7){$o$} \put(24,-7){$p$}\put(47,-7){$q$} \put(71,-7){$r$}
   \end{picture}
   \qquad\qquad
   \raisebox{-15pt}{%
    \begin{picture}(162,118)(-2,-8)
      \put(0,0){\includegraphics[height=110pt]{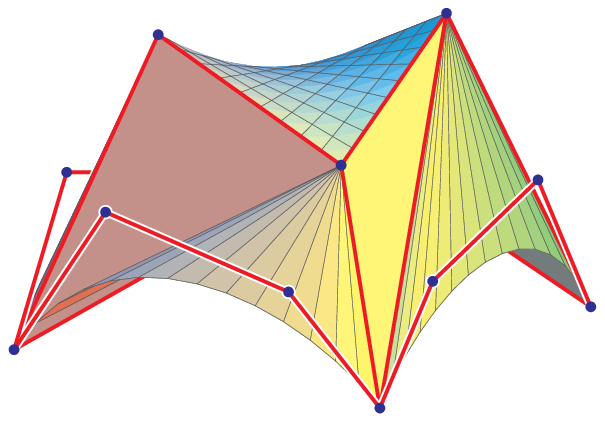}}
      \put(70,102){$A$}
      \put(0,90){$B$} \put(11,92){\vector(2,-1){40}}
      \put(44, -2){$C$} \put(52,9){\vector(1,2){10}}
     \put( 1,1){\vector(0,1){12}}     \put(27,1){\vector(0,1){45}}
     \put(75,1){\vector(0,1){29}}    
     \put(-2,-8){$o$} \put(24,-8){$p$}\put(72,-8){$q$} \put(99,-8){$r$}
    \end{picture}}}
 \end{equation}
The central quadrilateral $A$ in the decomposition corresponds to the bilinear 
patch at the top, the triangle $B$ in the decomposition corresponds to the indicated flat triangle,
and the triangle $C$ with points $o,p,q,r$ along one edge corresponds to the singular ruled cubic in
the surface.
The polygonal frame formed by the corresponding control points on the right is the control polygon
for this edge of $C$, which is a rational cubic B\'ezier curve.

We show that regular control surfaces are exactly
the possible limits of toric patches when the control points $\calB$
are fixed but the weights $w$ are allowed to vary. In particular,
the irregular control surface~\eqref{Eq:nonConvex} cannot be the
limit of toric B\'ezier patches.

Let $\lambda\colon\calA\to\R$ be a lifting function.
We use this and a given set of weights $w=\{w_\ba\in\R_>\mid\ba\in\calA\}$ to
get a set of weights which depends upon a parameter, 
$w_\lambda(t):=\{t^{\lambda(\ba)}w_\ba\mid \ba\in\calA\}$.
These weights are used to define a \demph{toric degeneration} of the patch,
\[
   F_{\calA,w,\calB,\lambda}(x;t)\ :=\ 
   \frac{\sum_{\ba\in\calA} t^{\lambda(\ba)}w_\ba \bb_\ba \beta_\ba(x)}%
        {\sum_{\ba\in\calA} t^{\lambda(\ba)}w_\ba\beta_\ba(x)}\,.
\]
Let $\calS_\lambda$ be the regular decomposition of $\calA$ induced by $\lambda$.
Our first main result is that the regular control surface $Y_{\calA,w,\calB}(\calS_\lambda)$ induced
by $\calS_\lambda$ is the limit of
the patches $Y_{\calA,w,\calB,\lambda}(t)$ parameterized by
$F_{\calA,w,\calB,\lambda}(x;t)$ as $t\to\infty$.

This limit is with respect to the Hausdorff distance between two subsets of
$\R^3$. 
Two subsets $X$ and $Y$ of $\R^3$ are \demph{within Hausdorff distance $\epsilon$} if for every
point $x$ of $X$ there is some point $y$ of $Y$ within a distance $\epsilon$ of $x$, and
vice-versa. 
With this notion of distance, we have the following result.

\begin{theorem}\label{Th:limit}
 ${\displaystyle \lim_{t\to\infty}Y_{\calA,w,\calB,\lambda}(t)=Y_{\calA,w,\calB}(\calS_\lambda)}$.
\end{theorem}

That is, for every $\epsilon>0$ there is a number $M$ such that if $t\geq M$, then the 
patch $Y_{\calA,w,\calB,\lambda}(t)$ and the regular control surface 
$Y_{\calA,w,\calB}(\calS_\lambda)$ are within Hausdorff distance $\epsilon$.

We illustrate Theorem~\ref{Th:limit} on a bicubic patch.
On the left below are the weights of a bicubic patch, in the center are the values
of a lifting function, and the corresponding regular decomposition is on the right.
\[
   \begin{picture}(63,64)
    \put(0,60){1}\put(20,60){3}\put(40,60){3}\put(60,60){1}
    \put(0,40){3}\put(20,40){9}\put(40,40){9}\put(60,40){3}
    \put(0,20){3}\put(20,20){9}\put(40,20){9}\put(60,20){3}
    \put(0, 0){1}\put(20, 0){3}\put(40, 0){3}\put(60, 0){1}
   \end{picture}
   \qquad \qquad
   \begin{picture}(69,64)
    \put(0,60){0}\put(20,60){2}\put(40,60){2}\put(60,60){0}
    \put(0,40){1}\put(20,40){1}\put(40,40){1}\put(60,40){1}
    \put(0,20){1}\put(20,20){2}\put(40,20){2}\put(60,20){1}
    \put(0, 0){0}\put(20, 0){1}\put(40, 0){1}\put(55, 0){0.5}
   \end{picture}
   \qquad \qquad
   \includegraphics{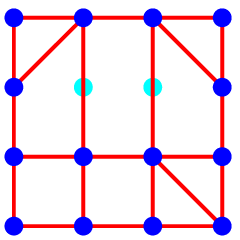}
\]
The two lighter points, $(1,2)$ and $(2,2)$, lie in no face of the decomposition.
Figure~\ref{Fig:Chungang} shows the toric degeneration of this bicubic patch at values $t=1$ and
$t=6$, and the regular control surface, all with the indicated control points.
\begin{figure}[htb]
\[
  \includegraphics[width=120pt]{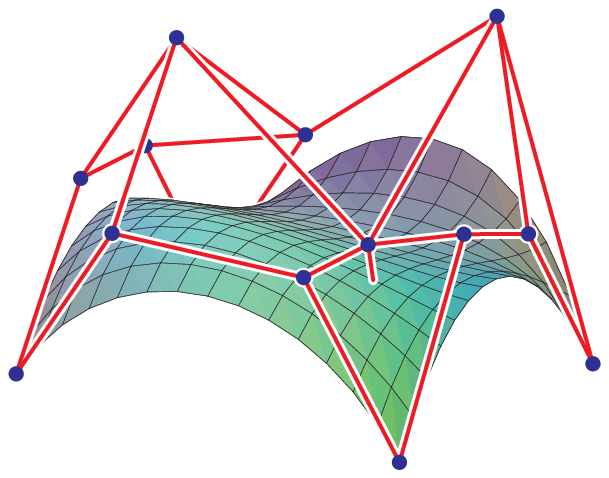}\qquad
  \includegraphics[width=120pt]{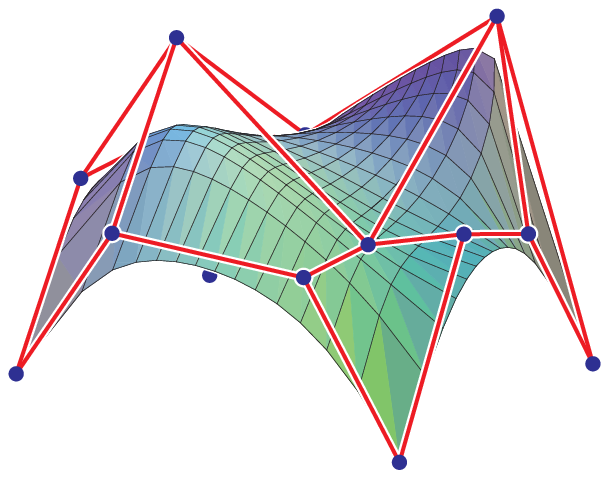}\qquad
  \includegraphics[width=120pt]{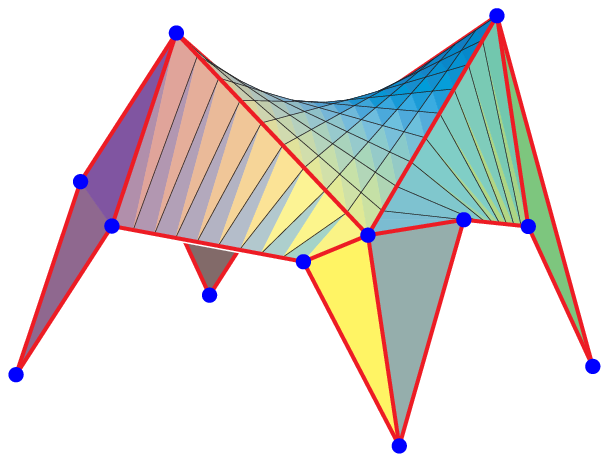}
\]
\caption{Toric degeneration of a rational tensor product patch of bidegree $(3,3)$.}
\label{Fig:Chungang}
\end{figure}

We will prove Theorem~\ref{Th:limit} in Appendix~\ref{Ap:one}.
The key idea is the factorization~\eqref{Eq:factorization} of the map
$F_{\calA,w,\calB,\lambda}(x;t)$ through the simplex $\simplex^\calA$.
This factorization allows us to study the limit in Theorem~\ref{Th:limit}
by considering the effect of the family of weights $w_\lambda(t)$
on the toric variety $X_{\calA}$ in $\simplex^\calA$.
Using equations for $X_{\calA}$, we can show the limit as $t\to\infty$ of  
the translated toric variety $X_{\calA,w_\lambda(t)}$ is a regular control surface in
$\R^\calA$ whose projection to $\R^3$ is the regular control surface
$Y_{\calA,w,\calB}(\calS_\lambda)$.

Figure~\ref{F:cubic_bezier} illustrates this lift, showing a toric degeneration of a rational cubic
B\'ezier curve, together with the corresponding degeneration of the curve
$X_{\calA,w}$ in the simplex $\simplex^\calA$.
Here, the weights are $w_\lambda(t)=(1,3t^2,3t^2,1)$.
That is, the control points $\bb_0$ and $\bb_3$ have weight 1, while the internal control points
$\bb_1$ and $\bb_2$ have weights $3t^2$.
\begin{figure}[ht]
\[
  \begin{picture}(135,192)(-2,0)
    \put(-1,10){\includegraphics[height=185pt]{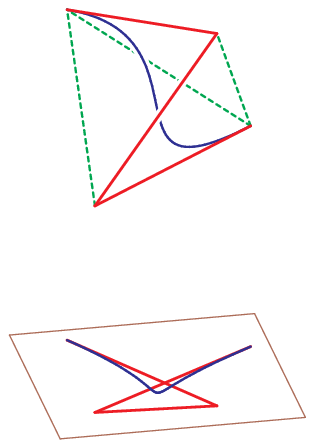}}
    \put(52,0){$t=1$}
    \put(5,105){$\pi_{\mathcal B}$} \put(21,135){\vector(0,-1){60}}
  \end{picture}
  \begin{picture}(133,187)
    \put(-1,10){\includegraphics[height=185pt]{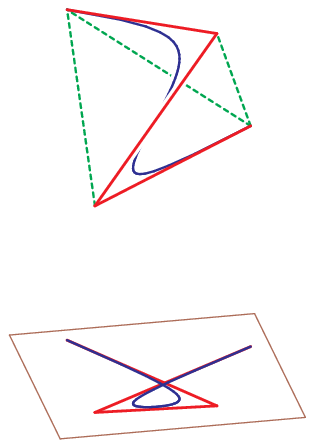}}
    \put(52,0){$t=3$}
    \put(5,105){$\pi_{\mathcal B}$} \put(21,135){\vector(0,-1){60}}
   \end{picture}
  \begin{picture}(133,187)
    \put(-1,10){\includegraphics[height=185pt]{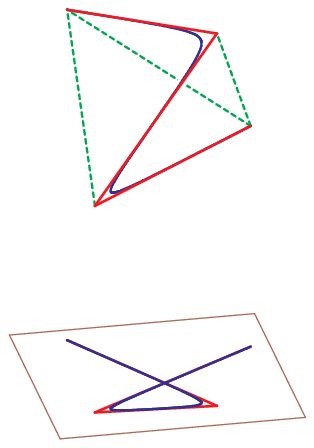}}
    \put(52,0){$t=9$}
    \put(5,105){$\pi_{\mathcal B}$} \put(21,135){\vector(0,-1){60}}
  \end{picture}
\]
\caption{Toric degenerations of a rational cubic B\'ezier curve.}
\label{F:cubic_bezier}
\end{figure}

By Theorem~\ref{Th:limit}, every regular control surface is
the limit of the corresponding patch under a toric degeneration.
Our second main result is a converse:   If a space $Y$ is the limit of patches of
shape $\calA$ with control points $\calB$, but differing weights, then $Y$ is a regular
control surface of shape $\calA$ and control points $\calB$.

\begin{theorem}\label{Th:convex}
  Let $\calA\subset\Z^2$ be a finite set and $\calB=\{\bb_\ba\mid\ba\in\calA\}\subset\R^3$ 
  a set of control points.
  If $Y\subset\R^3$ is a set for which there is a sequence $w^1,w^2,\dotsc$ of
  weights so that 
\[
    \lim_{i\to\infty} Y_{\calA,w^i,\calB}\ =\ Y\,.
\]
  then there is a lifting function $\lambda\colon\calA\to\R$ and a weight $w\in\R_>^\calA$
  such that $Y=Y_{\calA,w,\calB}(\calS_\lambda)$, a regular control surface.
\end{theorem}

To prove Theorem~\ref{Th:convex}, we consider the sequence of translated toric varieties
$X_{\calA,w^i}\subset\simplex^\calA$.
We show how work of  Kapranov, Sturmfels, and
Zelevinsky~\cite{KSZ1,KSZ2} implies that the set of all translated toric
varieties is naturally compactified by the set of all regular control surfaces in $\simplex^\calA$.
Thus some subsequence of $\{X_{\calA,w^i}\}$ converges to a regular control surface in
$\simplex^\calA$, whose image must coincide with $Y$, implying that $Y$ is a regular control
surface.
This method of proof does not give a simple way to recover a lifting function $\lambda$ or the
weight $w$ from the sequence of weights $w^1,w^2,\dotsc$.

We prove Theorem~\ref{Th:limit} in Appendix~\ref{Ap:one} and Theorem~\ref{Th:convex} in
Appendix~\ref{Ap:two}.
While both require more algebraic geometry than we have assumed so far,
Appendix~\ref{Ap:one} is more elementary and Appendix~\ref{Ap:two} is significantly more
sophisticated.

%
\appendix
\section{Proof of Theorem 5.1}\label{Ap:one}

Let $d,n$ be positive integers.
The definitions and results of Sections~\ref{S:toric}, \ref{S:regular}, and~\ref{S:control_surface},
as well as the statements of Theorems~\ref{Th:limit} and~\ref{Th:convex} make sense if we
replace $\calA\subset\Z^2$ by $\calA\subset\Z^d$ and $\calB\subset\R^3$ by
$\calB\subset\R^n$.
We work here in this generality.
This requires straightforward modifications such as replacing polygon by
polytope and in general removing restrictions on dimension.
We invite the reader to consult~\cite{Craciun} for a more complete treatment.

If the control points $\calB$ are the vertices
$\{e_\ba\mid\ba\in\calA\}\subset\R^\calA$ of the simplex
$\simplex^\calA$, then the toric patch $Y_{\calA,w,\calB}$ is 
the (translated) toric variety $X_{\calA,w}$. Given a decomposition $\calS$ of
$\calA$, write $X_{\calA,w}(\calS)$ for the control surface induced
by $\calS$ when the control points are the vertices of
$\simplex^\calA$. This is the union of patches $X_{\calF,w|_\calF}$
over all faces $\calF$ of $\calS$, and each patch
$X_{\calF,w|_\calF}$ lies in the face $\simplex^\calF$ of
$\simplex^\calA$ consisting of points $z$ whose coordinates $z_\ba$
vanish for $\ba\not\in\calF$.

Then, given any control points $\calB\subset\R^n$, we have
\[
   \pi_\calB(X_{\calA,w})\ =\ Y_{\calA,w,\calB}
   \qquad\mbox{and}\qquad
   \pi_\calB(X_{\calA,w}(\calS))\ =\ Y_{\calA,w,\calB}(\calS)\,.
\]
Because of this universality of $X_{\calA,w}$, $X_{\calA,w}(\calS)$, and the map
$\pi_\calB$, it suffices to prove Theorem~\ref{Th:limit} for limits of the toric variety
$X_{\calA,w}$.
Given a function  $\lambda\colon\calA\to\R$ and a weight $w\in\R^\calA_>$, define the family
of weights $\dcol{w_\lambda(t)}=\{t^{\lambda(\ba)}w_\ba\mid\ba\in\calA\}$ for $t\in\R_>$.

\begin{theorem}\label{Th:UniLimit}
 Let $\lambda\colon\calA\to\R$ be a lifting function and $\calS_\lambda$ the regular
 decomposition of $\calA$ induced by $\lambda$.
 Then, for any choice of weights $w\in\R^\calA_>$,
\[
   \lim_{t\to\infty} X_{\calA,w_\lambda(t)}\ =\ X_{\calA,w}(\calS_\lambda)\,.
\]
\end{theorem}

We prove Theorem~\ref{Th:UniLimit} in two parts.
We first show that any accumulation points of $\{X_{\calA,w_\lambda(t)}\mid t\geq 1\}$ as
$t\to\infty$ are
contained in the union of the faces $\simplex^\calF$ of the simplex $\simplex^\calA$
for each face $\calF$ of $\calS_\lambda$.
Then we show that $X_{\calF,w|_\calF}$ is the set of accumulation points contained in the face
$\simplex^\calF$, and in fact each accumulation point is a limit point of the sequence.
This will complete the proof of Theorem~\ref{Th:UniLimit} as
\[
    X_{\calA,w}(\calS_\lambda)\ =\ \bigcup_{\calF\in\calS} X_{\calF,w|_\calF}\,.
\]

We use homogeneous equations for the toric variety $X_{\calA,w}$.
Let $\mathbf{1}\in\R^\calA_>$ be the weight with every coordinate $1$.
Equations for $X_{\calA,\mathbf{1}}$ were described in~\cite[Prop.~B.3]{Craciun} as follows.
For every linear relation among the points of $\calA$ with nonnegative integer coefficients
 \begin{equation}\label{Eq:N-relation}
   \sum_{\ba\in\calA} \alpha_\ba \ba\ =\ \sum_{\ba\in\calA} \beta_\ba \ba\,
    \qquad\mbox{where}\qquad
   \sum_{\ba\in\calA}\alpha_\ba \ =\ \sum_{\ba\in\calA} \beta_\ba\,,
 \end{equation}
with $\alpha_\ba,\beta_\ba\in\N$,
we have the valid equation for points $z\in X_{\calA,\mathbf{1}}$,
 \begin{equation}\label{Eq:poly-relation}
   \prod_{\ba\in\calA}z_\ba^{\alpha_\ba}\ =\
   \prod_{\ba\in\calA}z_\ba^{\beta_\ba}\,.
 \end{equation}
Conversely, if $z\in\simplex^\calA$ satisfies equation~\eqref{Eq:poly-relation} for
every relation~\eqref{Eq:N-relation}, then $z\in X_{\calA,\mathbf{1}}$.
This follows from the description of toric ideals in~\cite[Ch.~4]{GBCP}.

Since the toric variety $X_{\calA,w}$ is obtained from $X_{\calA,\mathbf{1}}$ through
coordinatewise multiplication 
by $w=(w_\ba\mid \ba\in\calA)$, we have the following description of its equations.

\begin{proposition}\label{P:equations}
  A point $z\in \simplex^\calA$ lies in $X_{\calA,w}$ if and only if
 \[
   \prod_{\ba\in\calA}z_\ba^{\alpha_\ba} \cdot \prod_{\ba\in\calA}w_\ba^{\beta_\ba}
   \ =\
   \prod_{\ba\in\calA}z_\ba^{\beta_\ba} \cdot \prod_{\ba\in\calA}w_\ba^{\alpha_\ba}\,,
 \]
 for every relation~\eqref{Eq:N-relation} among the points of $\calA$.
\end{proposition}

\begin{rem}{\rm
 As every component of a point $z\in\simplex^\calA$ and weight $w$ is nonnegative, we may
 take arbitrary (positive) roots of the equations in Proposition~\ref{P:equations}.
 It follows that we may relax the requirement that the coefficients $\alpha_\ba$ and $\beta_\ba$
 in~\eqref{Eq:N-relation} are integers and allow them to be any nonnegative numbers such that
 \begin{equation}\label{Eq:convex-relation}
   \sum_{\ba\in\calA} \alpha_\ba \ba\ =\ \sum_{\ba\in\calA} \beta_\ba \ba\,
    \qquad\mbox{where}\qquad
   \sum_{\ba\in\calA}\alpha_\ba \ =\ \sum_{\ba\in\calA} \beta_\ba\ =\ 1\,,
 \end{equation}
 That is,
 $\sum \alpha_\ba\ba=\sum\beta_\ba\ba$ is a point in the convex hull of the set $\calA$
 that has more than one representation as a convex combination of points of $\calA$.
 Since $\calA\subset\Z^d$, we may assume that $\alpha_\ba,\beta_\ba$ are rational numbers.}
\end{rem}

Among all relations~\eqref{Eq:convex-relation} are those which arise when two subsets of
$\calA$ have intersecting convex hulls.

\begin{proposition}\label{Prop:toric_equations}
 Let $\calF,\calG\subset\calA$ be disjoint subsets whose convex hulls meet, 
\[
   \conv \calF \cap \conv\, \calG\ \neq\ \emptyset\,.
\]
 Then we have a relation of the form
\[
   \sum_{\ba\in\calF} \alpha_\ba \ba\ =\ \sum_{\ba\in\calG} \beta_\ba \ba\,
    \qquad\mbox{where}\qquad
   \sum_{\ba\in\calF}\alpha_\ba \ =\ \sum_{\ba\in\calG} \beta_\ba\ =\ 1\,,
\]
 with $\alpha_\ba,\beta_\ba\geq 0$.
 Thus
 \begin{equation}\label{Eq:toric_equations}
   \prod_{\ba\in\calF}z_\ba^{\alpha_\ba} \cdot \prod_{\ba\in\calG}w_\ba^{\beta_\ba}
   \ =\
   \prod_{\ba\in\calG}z_\ba^{\beta_\ba} \cdot \prod_{\ba\in\calF}w_\ba^{\alpha_\ba}\,,
 \end{equation}
 holds on $X_{\calA,w}$.
\end{proposition}

Given a subset $\calF\subset\calA$, the convex hull of the points $\{e_\bff\mid \bff\in\calF\}$
is the simplex \demph{$\bsimplex^\calF$}, which is a face of $\simplex^\calA$.
Under the tautological projection $\pi_\calA$ of $\simplex^\calA$ to $\Delta_\calA$,
the simplex $\simplex^\calF$ maps to the convex hull
$\Delta_\calF$ of $\calF$.
The \demph{geometric realization $|\calS|$} of a decomposition $\calS$ of $\calA$ is
the union of the simplices $\simplex^\calF$ for each face $\calF\in\calS$ of the
decomposition $\calS$.
We call a simplex $\simplex^\calF$ a \demph{face} of the geometric realization
$|\calS|$.
The images of the faces of the geometric realization $|\calS|$ under the
tautological projection $\pi_\calA$ form the faces of the polyhedral decomposition
$\calT(\calS)$.
Figure~\ref{F:Geom_real} illustrates this geometric realization for four regular decompositions of
$\calA=\{0,1,2,3\}$.
For this, $\simplex^\calA$ is
\begin{figure}[htb]
\[
  \begin{picture}(78,115)
    \put(0,15){\includegraphics[height=100pt]{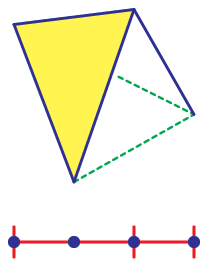}}
    \put(2,0){$\{0,1,2\}$, $\{2,3\}$}
   \end{picture}
   \qquad
   \begin{picture}(98,115)(-10,0)
    \put(0,15){\includegraphics[height=100pt]{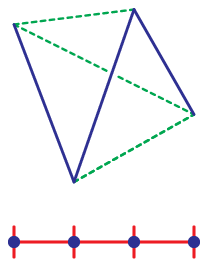}}
    \put(-10,0){$\{0,1\}$, $\{1,2\}$, $\{2,3\}$}
   \end{picture}
   \qquad
    \begin{picture}(78,115)
    \put(0,15){\includegraphics[height=100pt]{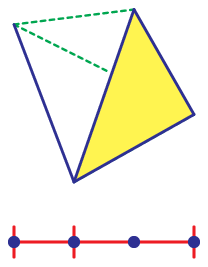}}
    \put(0,0){$\{0,1\}$, $\{1,2,3\}$}
   \end{picture}
   \qquad
    \begin{picture}(78,115)
    \put(0,15){\includegraphics[height=100pt]{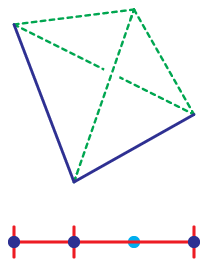}}
    \put(0,0){$\{0,1\}$, $\{1,3\}$}
   \end{picture}
\]
\caption{Geometric realizations for four decompositions.}
\label{F:Geom_real}
\end{figure}
the 3-dimensional simplex.
For each decomposition $\calS$ of $\calA$, we show the corresponding polyhedral decomposition of 
$\Delta_\calA=[0,3]$ and its facets.

Suppose that a point $z\in\simplex^\calA$ lies in the geometric realization
$|\calS|$ of a decomposition $\calS$ of $\calA$.
Then $z\in\simplex^\calF$ for some face $\calF$ of $\calS$, so that its \demph{support}
$\{\ba\in\calA\mid z_\ba\neq 0\}$ is a subset of $\calF$.
Conversely, any point $z\in\simplex^\calA$ whose support is a subset of some face $\calF$
of $\calS$ lies in $|\calS|$.
We conclude that $|\calS|\subset\simplex^\calA$ is the vanishing locus of the monomials
 \begin{equation}\label{Eq:SR-eqs}
   \{z_\ba\cdot z_\bb\mid \{\ba,\bb\}\not\subset \mbox{ any face $\calF$ of }\calS\}
    \ \cup\ \{z_\bc\mid \bc\not\in\mbox{ any face $\calF$ of }\calS\}\,.
 \end{equation}

A point $z\in\simplex^\calA$ is an \demph{accumulation point} of a sequence
$\{X_1,X_2,\dotsc\}$ of subsets of $\simplex^\calA$ if, for every $\epsilon>0$
and every $M$, there is some $m\geq M$ such that
$\dist(z,\, X_m)<\epsilon$.
Similarly, a point $z$ is an accumulation point of a family $\{X(t)\mid t\in\R_>\}$ if
for every $\epsilon>0$ and $M>0$, there is some $t>M$ such that
$\dist(z,\, X(t))<\epsilon$, and $z$ is a \demph{limit point} if 
$\lim_{t\to\infty}\dist(z,\, X(t))=0$.

\begin{lemma}\label{L:Accumulation}
 Let $w\in\R^\calA_>$ be a weight and $\lambda\colon \calA\to \R$ be a lifting
 function and $w_\lambda(t)$ the corresponding family of weights.
 Every accumulation point of $\{X_{\calA,w_\lambda(t)}\mid t\in\R_>\}$ lies in the
 geometric realization $|S_\lambda|$.
\end{lemma}

\begin{proof}
 We will show that a point $y\in\simplex^\calA$ which does not lie in $|\calS_\lambda|$
 cannot be an accumulation point of  $\{X_{\calA,w_\lambda(t)}\}$.
 If $y\in\simplex^\calA$ but $y\not\in|\calS_\lambda|$, then by~\eqref{Eq:SR-eqs} either there
 are points $\ba,\bb\in\calA$ with $y_\ba y_\bb\neq 0$ where $\{\ba,\bb\}$ do not lie in a
 common face of $\calS_\lambda$, or a single point $\bc\in\calA$ with $y_\bc\neq 0$ and $\bc$ does
 not participate in the decomposition $\calS_\lambda$.
 Set $\epsilon:=\min\{y_\ba,y_\bb\}$ (in the first case) or $\epsilon:=y_\bc$ (in the
 second case).
 We will show that if $t$ is sufficiently large and $z\in X_{\calA,w_\lambda(t)}$, then
 $\min\{z_\ba,z_\bb\}<\epsilon/2$ (in the first case) or $z_\bc<\epsilon/2$ (in the second
 case), which will complete the proof.

 Suppose that we are in the first case.
 Then the interior of the segment $\overline{\ba,\bb}$ meets some face $\Delta_\calF$ of
 $\calT(\calS_\lambda)$.
 If $\calF$ is the minimal such face, then the interiors of $\overline{\ba,\bb}$
 and $\Delta_\calF$ meet in a point $p$, and so we have the valid relation on $X_{\calA,w}$,
 \begin{equation}\label{Eq:special_toric}
   z_\ba^\mu z_\bb^\nu \cdot \prod_{\bff\in\calF}w_\bff^{\alpha_\bff}
   \ =\
   w_\ba^\mu w_\bb^\nu  \cdot  \prod_{\bff\in\calF}z_\bff^{\alpha_\bff} \,,
 \end{equation}
 by Proposition~\ref{Prop:toric_equations}, where
\[
   \dcol{p}\ :=\ \mu\ba+\nu\bb\ =\ \sum_{\bff\in\calF}\alpha_{\bff}\bff
   \qquad\mbox{and}\qquad
   \mu + \nu \ =\ 1\ =\ \sum_{\bff\in\calF}\alpha_{\bff}\,,
\]
 and the coefficients $\mu,\nu,\alpha_\bff$ are positive.
 For $X_{\calA,w_\lambda(t)}$  the relation~\eqref{Eq:special_toric} becomes
\[
   z_\ba^\mu z_\bb^\nu \cdot
   t^{\sum_{\bff\in\calF}\alpha_\bff\lambda(\bff)}\cdot
   \prod_{\bff\in\calF}w_\bff^{\alpha_\bff}
   \ =\
   \prod_{\bff\in\calF}z_\bff^{\alpha_\bff} \cdot
   t^{\mu\lambda(\ba) + \nu\lambda(\bb)} \cdot
   w_\ba^\mu w_\bb^\nu\,.
\]

Since the lift $\overline{(\ba,\lambda(\ba)),(\bb,\lambda(\bb))}$ of $\overline{\ba,\bb}$
does not lie on an upper facet of $P_\lambda$, but the lift of $\Delta_\calF$ does lie on
an upper facet, the point $p$ which is common to $\overline{\ba,\bb}$ and $\Delta_\calF$
is lifted lower on the lift of $\overline{\ba,\bb}$ than on the lift of
$\Delta_\calF$.
We thus have the inequality
 \begin{equation}\label{Eq:deltaDef}
   \mu\lambda(\ba) + \nu\lambda(\bb)\
   <\ \sum_{\bff\in\calF}\alpha_\bff\lambda(\bff)\,.
 \end{equation}
 Let $\delta>0$ be the difference of the two sides of~\eqref{Eq:deltaDef}.
 Then points $z\in X_{\calA,w(t)}$ satisfy
\[
    z_\ba^\mu z_\bb^\nu \ =\
   t^{-\delta}
   \prod_{\bff\in\calF}z_\bff^{\alpha_\bff} \cdot
   \frac{w_\ba^\mu w_\bb^\nu}{\prod_{\bff\in\calF}w_\bff^{\alpha_\bff}}
  \ <\  t^{-\delta} \cdot
   \frac{w_\ba^\mu w_\bb^\nu}{\prod_{\bff\in\calF}w_\bff^{\alpha_\bff}}\,,
\]
 as each component of $z\in\simplex^\calA$ is positive and at most 1.

 This inequality implies that if $t$ is sufficiently large, then at least one of the
 components $z_\ba,z_\bb$ is less than $\epsilon/2$, and thus $y$ is not an accumulation
 point of the sequence.
 A similar argument in the second case of $\bc\not\in\calS_\lambda$ completes the proof.
\end{proof}

We complete the proof of Theorem~\ref{Th:UniLimit} by showing that the set of accumulation
points of $X_{\calA,w_\lambda(t)}$ in $\simplex^\calF$ for $\calF$ a facet of
$\calS_\lambda$ is equal to $X_{\calF,w|_\calF}$, as this proves that
\[
   \lim_{t\to\infty} X_{\calA,w_\lambda(t)}\ =\ \bigcup_{\calF\in\calS_\lambda} X_{\calF,w|_\calF}\,.
\]

\begin{lemma}
 Let $\calF$ be a face of $\calS_\lambda$.
 Then $X_{\calF,w|_\calF}$ is the set of accumulation points of
 $\{X_{\calA,w_\lambda(t)}\mid t\in\R_>\}$ that lie in $\simplex^\calF$,
 and each point of $X_{\calF,w|_\calF}$ is a limit point.
\end{lemma}

\begin{proof}
 We have that $X_{\calF,w|_\calF}$ is the set of points $z\in\simplex^\calF$ such that
 \begin{equation}\label{Eq:facial}
   \prod_{\bff\in\calF}z_\bff^{\alpha_\bff} \cdot \prod_{\bff\in\calF}w_\bff^{\beta_\bff}
   \ =\
   \prod_{\bff\in\calF}z_\bff^{\beta_\bff} \cdot \prod_{\bff\in\calF}w_\bff^{\alpha_\bff}\,,
 \end{equation}
 whenever $\alpha,\beta\in\R_\geq^\calF$ satisfy
 \begin{equation}\label{Eq:affine_relation}
  \sum_{\bff\in\calF} \alpha_\bff \cdot \bff\ =\
  \sum_{\bff\in\calF} \beta_\bff \cdot \bff
   \qquad\mbox{where}\qquad
  \sum_{\bff\in\calF} \alpha_\bff \ =\
  \sum_{\bff\in\calF} \beta_\bff  \ =:\ m\,.
 \end{equation}
 The equation~\eqref{Eq:facial} also holds on $X_{\calA,w}$, and on
 $X_{\calA,w_\lambda(t)}$, it becomes
 \begin{equation}\label{Eq:lambda_facial}
   \prod_{\bff\in\calF}z_\bff^{\alpha_\bff} \cdot \prod_{\bff\in\calF}w_\bff^{\beta_\bff}
   \cdot t^{\sum_{\bff}\beta_\bff \cdot \lambda(\bff)}
   \ =\
   \prod_{\bff\in\calF}z_\bff^{\beta_\bff} \cdot \prod_{\bff\in\calF}w_\bff^{\alpha_\bff}
   \cdot t^{\sum_{\bff}\alpha_\bff \cdot \lambda(\bff)}\,,
 \end{equation}
 Observe that
\[
   \frac{1}{m} \sum_{\bff\in\calF} \alpha_\bff \cdot \bff \ =\
   \frac{1}{m} \sum_{\bff\in\calF} \beta_\bff \cdot \bff
\]
 is a point in the convex hull of $\calF$.
 Since $\calF$ is a face of the decomposition induced by $\lambda$, the function $\lambda$
 is affine-linear on $\calF$ and so
\[
   \sum_{\bff\in\calF} \alpha_\bff \cdot \lambda(\bff) \ =\
   \sum_{\bff\in\calF} \beta_\bff \cdot  \lambda(\bff)\ .
\]
 Let this common value be $\delta$.
 As dividing~\eqref{Eq:lambda_facial} by $t^\delta$ gives~\eqref{Eq:facial}, we see
 that~\eqref{Eq:facial} is also a valid relation on every
 member of the family $\{X_{\calA,w_\lambda(t)}\mid t\in\R_>\}$, whenever
 $(\alpha_\bff,\beta_\bff\mid\bff\in\calF)$ satisfy~\eqref{Eq:affine_relation}.
 It follows that this set of equations~\eqref{Eq:facial} holds on every accumulation point
 of the family $\{X_{\calA,w_\lambda(t)}\mid t\in\R_>\}$, which implies that those
 accumulation points lying in $\simplex^\calF$ are a subset of $X_{\calF,w|_\calF}$.

 To show the other inclusion, for each face $\calF$ of $\calS_\lambda$, let
 \demph{$X_{\calF,w|_\calF}^\circ$} consist
 of those points $z\in X_{\calF,w|_\calF}$ with $z_\bff\neq 0$ for $\bff\in\calF$.
 Evidently we have
\[
    X_{\calA,w}(\calS_\lambda)\ =\
    \coprod_{\calF\in\calS_\lambda}X_{\calF,w|_\calF}^\circ\,,
\]
 and so it suffices to prove that every point of $X_{\calF,w|_\calF}^\circ$ is a
 limit point of the family $\{X_{\calA,w_\lambda(t)}\mid t\in\R_>\}$.

 Since $\calF$ is a face of $\calS_\lambda$, there is a vector $\bv\in\R^d$ such that the
 function $\calA\to\R$,
\[
   \ba\ \longmapsto\ \bv\cdot\ba+\lambda(\ba)
\]
 is maximized on $\calF$ with maximum value $\delta$.
 That is, if $\bv\cdot\ba+\lambda(\ba)\geq\delta$ with $\ba\in\calA$,
 then $\ba\in\calF$ and $\bv\cdot\ba+\lambda(\ba)=\delta$.

 Consider the action of $t\in\R_>$ on $x\in\R^d_>$ where
\[
   (t*x)_i\ :=\ t^{v_i}x_i\,.
\]
 Let $z\in X_{\calF,w|_\calF}^\circ$.
 Then $z=\varphi_{\calF,w}(x)$ for some $x\in\R^d_>$, and we have
\[
   \varphi_{\calA,w}(t*x)_\ba\ =\ w_\ba\cdot t^{\bv\cdot\ba}x^\ba\,,
\]
 and so, under the action $(t.z)_\ba=t^{\lambda(\ba)} z_\ba$ of $\R_>$ on $\R^\calA_>$, we
 have
\[
  t.\varphi_{\calA,w}(t*x)_\ba\ =\ w_\ba\cdot t^{\bv\cdot\ba+\lambda(\ba)}x^\ba\,.
\]
 Then the line through $t.\varphi_{\calA,w}(t*x)$ is equal to the line through
\[
    t^{-\delta}\bigl(t.\varphi_{\calA,w}(t*x)\bigr)\,,
\]
 whose $\ba$-coordinate is
\[
   w_\ba\cdot t^{\bv\cdot\ba+\lambda(\ba)-\delta}x^\ba\,.
\]
 Since $\bv\cdot\ba+\lambda(\ba)-\delta\leq 0$ with equality only when $\ba\in\calF$,
 we see that the limit of the points $t.\varphi_{\calA,w}(t*x)$ of $\Delta^\calA$
 as $t\to\infty$ is the point $\varphi_{\calF,w|_\calF}(x)=z$,
 which completes the proof.
\end{proof}

%
\section{Proof of Theorem 5.2}\label{Ap:two}

Theorem~\ref{Th:UniLimit} shows that a limit of translates of $X_{\calA,w}$ by a
one-parameter subgroup of $\R^\calA_>$ (a toric degeneration of $X_{\calA,w}$) is a regular control
surface. 
This is a special and real-number case of more general results of Kapranov, Sturmfels, and
Zelevinsky~\cite{KSZ1,KSZ2} concerning all possible toric degenerations of the
complexified toric variety $X_{\calA}(\C)$ that we will use to prove
Theorem~\ref{Th:convex}.

Suppose the $\calA\subset\Z^d$ is a finite set of integer lattice
points. We will assume that $\calA$ is primitive in that differences
of elements of $\calA$ span $\Z^d$:
\[
   \Z^d\ =\    \Z\langle \ba-\ba'\mid \ba,\ba'\in\calA \rangle\,,
\]
that is, $\calA$ affinely spans $\Z^d$ (if not, then simply replace $\Z^d$ by the affine span
of $\calA$).
Let $\P^\calA$ be the complex projective space with homogeneous coordinates $[z_\ba\mid \ba\in\calA]$
indexed by elements of $\calA$.
(These are extensions to all of $\P^\calA$ of the homogeneous coordinates~\eqref{Eq:Homog_Coords},
 which were valid for the nonnegative part of $\P^\calA$.)
The complex torus $H:=(\C^*)^d$ naturally acts on $\P^\calA$ with weights given by the set $\calA$:
$t\in H$ sends the point $z$ with homogenous coordinates $[z_\ba\mid \ba\in\calA]$ to the point 
 $t.z:=[t^\ba z_\ba\mid \ba\in\calA]$.
Note that $X_{\calA,\mathbf{1}}(\C)$ is the closure of the orbit of $H$ through the point 
$\dcol{\mathbf{1}}:=[1:\dotsc:1]$.
For $w\in(\C^*)^\calA$, the translate $w.X_{\calA,\mathbf{1}}(\C)=:\dcol{X_{\calA,w}(\C)}$ is also the
closure of the orbit of $H$ through the point $w$ (considered as a point in $\P^\calA$).
Note that $\simplex^\calA$ is the nonnegative real
part~\cite[Ch.~4]{Fu93} of $\P^\calA$, and when
$w\in\R^\calA_>\subset(\C^*)^\calA$, then $X_{\calA,w}$ is the nonnegative real part of
$X_{\calA,w}(\C)$.

A \demph{toric degeneration} of $X_{\calA,\mathbf{1}}(\C)$ is any translate $X_{\calA,w}(\C)$, or
any limit of translates
\[
   \lim_{t\to 0} \lambda(t).X_{\calA,w}(\C) 
\] 
where $\lambda\colon \C^*\to (\C^*)^\calA$ is a one-parameter subgroup.
This limit is the same limit as in Section~\ref{Ap:one}, its ideal is the limit of the ideals of 
$\lambda(t).X_{\calA,w}(\C)$ as $t\to 0$.
We remark that the data of a one-parameter subgroup of $(\C^*)^\calA$ are equivalent to  
 homomorphisms of abelian groups $\Z^\calA\to\Z$ and thus to functions $\lambda\colon\calA\to\Z$,
which explains our notation $\lambda$.

The translates $X_{\calA,w}(\C)$ for $w\in(\C^*)^\calA$ give a
family of subvarieties of $\P^\calA$, each with the same dimension
and degree, and each equipped with an action of $H$. A main result
of~\cite{KSZ1,KSZ2} identifies all suitable limits of these
translates $X_{\calA,w}(\C)$ with the points of a complex
projective toric variety \dcol{$C_\calA(\C)$}. 
The points of $C_\calA(\C)$
in turn are in one-to-one correspondence with all possible complex
toric degenerations of $X_{\calA,w}(\C)$ as $w$ ranges over
$(\C^*)^\calA$. For a toric degeneration $X$ of a translate of
$X_{\calA,\mathbf{1}}(\C)$, we write \dcol{$[X]$} for the corresponding point of
$C_\calA(\C)$. We will use this result to prove
Theorem~\ref{Th:convex} as follows.

\begin{proof}[Proof of Theorem~$\ref{Th:convex}$]
 Fix control points $\calB=\{\bb_\ba\mid\ba\in\calA\}\subset\R^n$ and suppose that
 $\{w^1,w^2,\dotsc\}$ is a sequence of weights in $\R^\calA_>$ such that the sequence of
 toric patches $\{Y_{\calA,w^i,\calB}\mid i=1,2,\dotsc\}$ converges to a set $Y$ in
 $\R^n$ in the Hausdorff topology.

 Consider the corresponding sequence $\{X_{\calA,w^i}(\C)\mid i=1,2,\dotsc\}$ of torus
 translates of $X_{\calA}(\C)$.
 This gives a sequence $[X_{\calA,w^i}(\C)]$ of points in the projective
 toric variety $C_{\calA}(\C)$.
 Since $C_{\calA}(\C)$ is compact, this sequence of points has a convergent subsequence
 whose limit point is a toric degeneration 
\[
  \lim_{t\to\infty} X_{\calA,w_\lambda(t)}(\C)\ =\ 
   X_{\calA,w}(\calS_\lambda)(\C)\ =\ 
   \bigcup_{\calF\in\calS_\lambda} X_{\calF,w|_\calF}(\C)
\]
 of $X_{\calA,w}(\C)$ for some
 $w\in(\C^*)^\calA$ and lifting function $\lambda\colon\Z^\calA\to\Z$.
 Replacing the original weights $\{w^i\}$ by this subsequence, we may assume that, as
 points of $C_\calA(\C)$, we have
 \begin{equation}\label{Eq:Chow_limit}
   \lim_{i\to\infty} [X_{\calA,w^i}(\C)]\ =\
   [X_{\calA,w}(\calS_\lambda)(\C)]\ =\
   \lim_{t\to\infty}\, [X_{\calA,w_\lambda(t)}(\C)]\,.
 \end{equation}
 The points $[X_{\calA,w^i}(\C)]$ of $C_\calA(\C)$ are translates of the base point
 $[X_\calA(\C)]$ by elements of $\R_>^\calA\subset(\C^*)^\calA$, and so they lie in the
 nonnegative real part of the toric variety $C_\calA(\C)$,  and therefore so does their
 limit point.
 But by~\eqref{Eq:Chow_limit} this limit point is a translate of $X_\calA(S_\lambda)(\C)$,
 and thus it is a translate by a real weight.
 This shows that we may take the weight $w$ in~\eqref{Eq:Chow_limit} to be real.

 Theorem~\ref{Th:convex} will follow from this and the claim that if 
 a sequence $\{[X_i]\mid i\in\N\}\subset C_\calA(\C)$ converges to a point 
 $[X]$ of $C_\calA(\C)$ in the (usual) analytic topology, then the
 sequence of corresponding subvarieties $\{X_i\}$ converges to $X$ in the Hausdorff metric on
 subsets of $\P^\calA$.
 Given this claim,~\eqref{Eq:Chow_limit} implies that in the Hausdorff topology on subsets of
 $\P^\calA$, 
\[
  \lim_{i\to\infty} X_{\calA,w^i}(\C)\ =\
  \lim_{t\to\infty} X_{\calA,w_\lambda(t)}(\C)\ =\
  X_{\calA,w}(S_\lambda)(\C)\,.
\]
 We may restrict this to their real points to conclude that the limit of patches
 $X_{\calA,w^i}$ in $\simplex^\calA$,
\[
  \lim_{i\to\infty} X_{\calA,w^i}\ =\
  \lim_{t\to\infty} X_{\calA,w_\lambda(t)}\ =\
  X_{\calA,w}(S_\lambda)\,,
\]
 is a regular control surface.
 Since $Y_{\calA,w^i,\calB}=\pi_\calB(X_{\calA,w^i})$, the limit
 $\lim_{i\to\infty}Y_{\calA,w^i,\calB}$ equals
\[
  \lim_{i\to\infty} \pi_\calB(X_{\calA,w^i})\ =\
  \pi_{\calB}\bigl(\lim_{i\to\infty} X_{\calA,w^i}\bigr)\ =\
  \pi_{\calB}\bigl(X_{\calA,w}(\calS_\lambda)\bigr)\ =\
  Y_{\calA,w,\calB}(\calS_\lambda)\,,
\]
 which is a regular control surface.
 This will complete the proof of Theorem~\ref{Th:convex}, once we have proven the claim.
\end{proof}

\begin{proof}[Proof of claim]
 As shown in~\cite{KSZ1,KSZ2}, the projective toric variety $C_\calA(\C)$ is the Chow quotient
 of $\P^\calA$  by the group $H=(\C^*)^d$ acting via the weights of $\calA$.
 We explain this construction.
 Let $D=d!\cdot\vol(\Delta_\calA)$, which is the degree of the projective toric variety
 $X_\calA(\C)$, as well as any of its translates.
 Basic algebraic geometry (see~\cite[Lect.~21]{Harris}) gives us the existence of a complex
 projective variety \dcol{$C(D,d,\calA)$}, called the \demph{Chow variety}, whose points
 are in one-to-one correspondence with $d$-dimensional cycles in $\P^\calA$ of degree $D$.
 These are formal linear combinations
 \begin{equation}\label{Eq:cycle}
   Z\ :=\ \sum_{i=1}^m D_j Z_j
 \end{equation}
 where each coefficient $D_j$ is a nonnegative integer, each $Z_j$ is a reduced and
 irreducible subvariety of $\P^\calA$ of dimension $d$, and we have
\[
   D\ =\ \sum_{j=1}^m D_j\cdot \degree(Z_j)\,.
\]
 In particular all translates $X_{\calA,w}(\C)$ are represented by points of
 $C(D,d,\calA)$.

 The torus $(\C^*)^\calA$ (in fact $(\C^*)^\calA/H$) acts on $\P^\calA$ and thus on
 $C(D,d,\calA)$, with the points representing the translates $X_{\calA,w}$ forming a single orbit.
 The \demph{Chow quotient $C_\calA(\C)$} is the closure of this orbit in $C(D,d,\calA)$.
 An explicit description of $C_\calA(\C)$ in terms of a quotient fan or the secondary
 polytope of $\calA$ is more subtle and may be found in~\cite{KSZ1,KSZ2}.
 We do not need this description to prove Theorem~\ref{Th:convex}, although such a
 description could be used to help identify the limit control surface whose existence is
 only asserted by Theorem~\ref{Th:convex}.

 The points of $C_\calA(\C)$ correspond to toric degenerations of translates
 $X_{\calA,w}(\C)$.
 We describe this, associating a cycle of degree $D$ and dimension $d$ to any toric
 degeneration.
 Let $\lambda\colon\Z^\calA\to\Z$ be any lifting function with corresponding regular
 decomposition $\calS_\lambda$ of $\calA$.
 Let $\calF\subset\calA$ be a facet of $\calS_\lambda$.
 Set \dcol{$\delta_\calF$} to be the index in $\Z^d$ of the lattice
\[
   \Z\langle \bff-\bff' \mid \bff,\bff'\in\calF\rangle
\]
 spanned by differences of elements of $\calF$.
 Then $X_{\calF,w|_\calF}(\C)$ is a subvariety of $\P^\calF$ of dimension $d$ and degree
\[
   d!\cdot \vol(\Delta_\calF)/\delta_\calF\,.
\]
 The toric degeneration
\[
   \lim_{t\to\infty} X_{\calA,w_\lambda(t)}(\C)\ =\
   \bigcup_{\calF\mbox{ a facet of }\calS_\lambda}X_{\calF,w|_\calF}(\C)
\]
 (this is a set-theoretic limit) corresponds to the cycle
\[
   \sum_{\calF\mbox{ a facet of }\calS_\lambda} \delta_\calF X_{\calF,w|_\calF}(\C)\,,
\]
 which has degree $D=d!\cdot\vol(\Delta_\calA)$.

 We now prove the claim.
 Following Lawson~\cite[Sect.~2]{Lawson}, Kapranov, Sturmfels, and
 Zelevinsky~\cite[Sect.~1]{KSZ1} associate to a cycle~\eqref{Eq:cycle} a current on
 $\P^\calA$---the linear functional $\int_Z$ of integrating a smooth $2d$-form over the cycle $Z$.
 The analytic topology on the Chow variety is equivalent to the weak topology on 
 currents.
 (The weak topology is the topology of pointwise convergence: A sequence $\{\psi_i\mid i\in\N\}$
 of currents converges to a current $\psi$ if and only if for every $2d$-form $\omega$ on
 $\P^\calA$ we have $\lim_{i\to\infty} \psi_i(\omega)=\psi(\omega)$, as complex numbers.)

 Suppose that $\{[X_i]\}\subset C_\calA(\C)$ is a convergent sequence of points in the usual
 analytic topology on $C_\calA(\C)$, with limit point $[X]$,
\[
    \lim_{i\to \infty} [X_i]\ =\ [X]\,.
\]
 Then the associated currents converge.
 That is, for every smooth $2d$-form $\omega$, we have
 \begin{equation}\label{Eq:int}
   \lim_{i\to\infty}\int_{X_i}\,\omega\ =\ \int_X\,\omega\,.
 \end{equation}
 We use this to show that $\lim_{i\to\infty} X_i=X$, in the Hausdorff metric.

 Given a point $x\in X$ and a number $\epsilon>0$, let $\omega$ be a $2d$-form with
 $\int_X\omega\neq 0$ which vanishes outside the ball $B(x,\epsilon)$ of radius $\epsilon$ around $x$.
 Then~\eqref{Eq:int} implies that there is a number $M$ such that if $i>M$, then 
 $\int_{X_i}\omega\neq 0$, and  thus $X_i\cap B(x,\epsilon)\neq\emptyset$.
 Since $X$ is compact, there is some number $M$ such that if $i>M$, then every point of $X$ 
 is within a distance $\epsilon$ of a point of $X_i$.

 To complete the proof of the claim, we need to show that for every number $\epsilon>0$, there
 is a number $M$ such that if $i>M$, then every point of $X_i$ lies within a distance $\epsilon$ of
 $X$.
 Suppose that this is not true.
 That is, there is a number $\epsilon>0$ such that for every number $M$,
 there is some $i>M$ such that $X_i$ has a point whose distance from $X$ exceeds
 $\epsilon$.
 Replacing $\{X_i\}$ by a subsequence, we may assume that each $X_i$ has a point $x_i$ whose distance from
 $X$ exceeds $\epsilon$.
 It is no loss to assume that the points $x_i$ are smooth.
 By the compactness of $\P^\calA$ and of the Grassmannian of $d$-dimensional linear subspaces of 
 $\P^\calA$, we may replace $\{X_i\}$ by a subsequence and assume that the points $x_i$ converge to
 a point $x$, and that the tangent spaces $T_{x_i}X_i$ also converge to a linear space $L$. 
 It follows that there is a $2d$-form $\omega$ which vanishes outside of $B(x,\epsilon/2)$
 and $\int_L\omega\neq 0$.
 By our assumption on the sequence of tangent spaces $T_{x_i}X_i$, we will have
\[
   \lim_{i\to \infty} \int_{X_i} \omega\ \neq\  0\,.
\]
 But then~\eqref{Eq:int} implies that $\int_X\omega\neq 0$, and so
 $X\cap B(x,\epsilon/2)\neq\emptyset$, which contradicts our assumption that $X$ is 
 the limit the spaces $X_i$.

 This completes our proof of the claim, and therefore of Theorem~\ref{Th:convex}.
\end{proof}
\def\cprime{$'$}
\providecommand{\bysame}{\leavevmode\hbox to3em{\hrulefill}\thinspace}
\providecommand{\MR}{\relax\ifhmode\unskip\space\fi MR }
\providecommand{\MRhref}[2]{%
  \href{http://www.ams.org/mathscinet-getitem?mr=#1}{#2}
}
\providecommand{\href}[2]{#2}

\end{document}